\documentclass[english, conference]{IEEEtran}

\usepackage{url}
\usepackage{graphicx, epsfig}
\usepackage{amsmath}
\usepackage{subfig}
\usepackage{babel}
\newtheorem{lemma}{Lemma}
\begin{document}
\title{Improving DTN Routing Performance Using Many-to-Many Communication: A Performance Modeling Study}
	
	\author{\IEEEauthorblockN{Giridhari Venkatadri, V.~Mahendran, and C.~Siva Ram Murthy}
\IEEEauthorblockA{Department of Computer Science and Engineering\\
Indian Institute of Technology Madras, India\\
Email: {giridhariv@gmail.com, mahendra@cse.iitm.ac.in, and murthy@iitm.ac.in}}}
	\maketitle
	
	\begin{abstract}
	Delay-Tolerant Networks~(DTNs) have emerged as an exciting research area with a number of useful applications. Most of these applications would benefit greatly by a reduction in the message delivery delay experienced in the network. The delay performance of DTNs is adversely affected by contention, especially severe in the presence of higher traffic rates and node densities. Many-to-Many (M2M) communication can handle this contention much better than traditional one-to-one communication employing CSMA. In this paper, for the first time, we analytically model the expected delivery delay of a DTN employing epidemic routing and M2M communication. The accuracy of our model is demonstrated by matching the analytical results against those from simulations. We also show using simulations that M2M communication significantly improves the delay performance (with respect to one-to-one CSMA) for high-contention scenarios. We believe our work will enable the effective application of M2M communication to reduce delivery delays in DTNs.

	\end{abstract}
	
	\section{Introduction}
        Delay-Tolerant Networks~(DTNs)~\cite{dtnarch},~\cite{dtnsurvey} are challenged networks where end-to-end connectivity does not always exist between nodes. One reason this may occur is that links get disrupted either due to interference or due to nodes moving out of range of each other. DTNs can be put to a variety of applications such as bringing Internet connectivity to villages, and reducing the pressure on cellular bandwidth by offloading delivery through WiFi and mobiles of people moving in the vicinity. Since an end-to-end path seldom exists between source and destination, most existing Internet protocols fail and instead, the store-carry-forward~\cite{dtnarch} protocol is used to transfer messages from source to destination through a number of intermediate relays. The lack of an end-to-end path means message delivery delays are large. Hence, lowering the message delivery delays would be critical in enhancing the performance of most of these applications.

	It has been shown that contention adversely affects the delay performance of DTNs~\cite{contentionaware}. This is because contention implies that transfer of a message is not always possible when one node meets another, owing to two of the following reasons:~$(i)$~the message contends with other messages in the buffer of the node to be chosen for transmission, and~$(ii)$~the node contends with other nodes in the vicinity for opportunities to communicate.

While contention has been recognized as an important factor that increases delivery delays, the problem of fighting contention and thereby reducing delivery delays in DTNs has not been adequately considered. For a lot of DTN application scenarios, it is possible for a number of nodes to come in the vicinity of each other. For example, people in a crowded public square, or villagers near a road where a vehicle is passing by. In these scenarios, as a result of contention, most nodes would not be able to transmit the messages they carry since existing access schemes such as Carrier Sense Multiple Access~(CSMA) are one-to-one communication schemes and would allow only one pair of nodes within the region of contention to communicate at any instant of time.

Many-to-Many~(M2M) communication~\cite{m2mtwc09},~\cite{m2minfocom} has been proposed to improve the capacity of Mobile Ad hoc NETworks~(MANETs) and enables each node in a bounded group of nodes to simultaneously communicate with each other node in the group. Hence when a number of nodes come in the vicinity of each other, this precious contact duration is used much better, with more nodes getting a chance to communicate. Hence, M2M communication can be used to battle the effects of contention and therefore reduce message delivery delays. M2M communication has become feasible since the state-of-the-art technology enables support for both access schemes (such as Code Division Multiple Access~(CDMA)) and location services such as Global Positioning System~(GPS)) on a single IC chip~\cite{m2mtwc09}.

In this paper, we analytically model the expected delivery delay of a DTN employing epidemic routing and M2M communication. Despite the fact that the analysis has been done for epidemic routing, the framework we develop is easily extended to other routing protocols. The accuracy of our model is demonstrated by matching the results from analysis against those from simulations. Also, simulations show that for high-contention scenarios, M2M communication significantly outperforms one-to-one communication employing CSMA in terms of delivery delay. 

\section{Related Work}
Initial analytical models developed for DTN routing performance study~\cite{pqepidemic},~\cite{spraynwait},~\cite{multicopy} worked under the assumption that whenever two nodes are in contact with each other, all messages could always be successfully transferred from one node to the other~($i.e.,$~they assumed both buffer capacity and bandwidth to be infinite). While papers such as~\cite{ode} have modeled DTN performance with bounded buffer capacity, they have assumed that infinite bandwidth is available and hence that there is no contention.
The motivation for not considering contention has been that DTNs are sparse networks and such sparsity yields negligible contention. However, this conjecture has been disproved with the help of simulations in works such as~\cite{spraynwait},~\cite{multicopy}. The authors in~\cite{contentionaware} show via simulations that, irrespective of whether the network is sparse or dense, the contention is substantial for high traffic rates; and also that the contention increases with an increase in network density. Realizing the importance of contention in the routing performance, the authors have attempted to include contention in the analysis of routing~\cite{contentionaware},~\cite{contentionaware-conf}. Their analysis assumes the one-to-one CSMA communication scheme.

In our paper, however, we propose that in scenarios where there are areas of high node density (\emph{i.e.,}~a number of nodes in the vicinity of each other), nodes must co-operate using an M2M communication scheme such as the one proposed in~\cite{m2mtwc09} and not compete as assumed by~\cite{contentionaware}. This would mean that instead of just one pair of nodes in such an area of higher node density communicating, we can have a larger set of nodes communicating with each other at the same time. This would significantly increase the probability that two nodes which come into contact with each other get to exchange messages, and thereby reduce the routing delay. The beneficial effect of using such an M2M communication scheme would be amplified in the presence of high traffic loads.

\section{Analytical Model}
In this section, we develop an analytical model for the expected delivery delay of a DTN employing M2M communication. We follow the delivery of a particular message from source node to the destination node, and model the evolution of the state of the system with respect to this message. We first describe the Medium Access Control~(MAC) protocol used for achieving M2M communication which logically divides the network into square cells (or) tiles. We then derive the mobility statistics for the Random Direction~(RD) mobility model, where the network is logically divided into square tiles as described in the MAC protocol. Using these statistics, we model the evolution of the network as a stochastic process and thereby derive the expected delivery delay.

\subsection{MAC Protocol for M2M Communication}\label{sec:MACProtocol}
We adopt the scheme proposed by Moraes et. al., in~\cite{m2mtwc09} for M2M communication in MANETs. While they have proposed this scheme to increase the capacity of MANETs, we adopt the scheme to reduce delivery delays in DTNs. For the sake of completeness, we briefly describe the scheme in this section. The network is logically divided into a number of square tiles. Each node, with the help of GPS, is aware of its current location (thereby the current tile it is in). GPS also helps in time synchronization among the nodes. Therefore, the communication in each tile is synchronized. In each tile, communication is performed in the form of sessions that have two phases. The first phase, called the discovery phase, is divided into slots. Each node in the grid that wants to communicate randomly picks one slot and broadcasts its presence to the other nodes in that slot, along with some other control plane information. The first bunch of nodes~(say $\alpha$-nodes) that successfully do so enter the second phase, which is the communication phase. The parameter $\alpha$ is chosen according to different factors such as the allowable receiver complexity. We assume there are sufficient slots to make the probability of collision negligible thereby ensuring that if there are lesser than $\alpha$ nodes in a particular cell, then all of them are chosen for communication with a high probability. Also, this means that when there are at least $\alpha$ nodes in a particular cell, the MAC scheme effectively chooses a random subset of $\alpha$-nodes from the set of nodes in that cell.
Each node that enters the communication phase is assigned a code on which it transmits data and a frequency on which it receives data. Using narrowband FDMA/CDMA, all $\alpha$ nodes are able to communicate with each other simultaneously ($i.e.,$~each node participating in the communication session can send up to $b_{BW}$ messages to every other node participating in the session, where $b_{BW}$ depends on the bandwidth available for M2M communication).
\begin{figure*}[th]
\begin{equation}\label{lengthyeqtn}
\begin{split}
  \overline{n} =  \int_{t=0}^{\infty}\int_{\theta=0}^{\frac{\pi}{2}}\int_{v=V_{min}}^{V_{max}}\int_{y=0}^{a}\int_{x=0}^{a}\biggl( 1 + \frac{vt\left(cos\theta + sin\theta \right)}{a} + \frac{x+y}{a}\biggr)\, 
\left(\frac{1}{a}\right)dx\left(\frac{1}{a}\right)dy\left(\frac{1}{V_{max}-V_{min}}\right)dv\\\left(\frac{1}{\frac{\pi}{2}}\right)d\theta \left(\frac{\overline{v}}{\overline{L}}\right)e^{-\left(\frac{\overline{v}}{\overline{L}}\right)t}dt
\end{split}
\end{equation}
\begin{center}
\line(1,0){510}
\end{center}
\end{figure*}

\subsection{Mobility Statistics}
\subsubsection{Intermeeting Time (IMT)}
The IMT of any two nodes is defined as the time elapsed between two successive meetings of the nodes. Since we logically tile the terrain to make the MAC easier, the IMT is expected to increase since some communication opportunities are lost as a result of tiling the terrain. Hence existing derivations~\cite{lambdard} do not hold as such. We adopt the methodology used in~\cite{methodology} to derive the distribution of the IMT, show it to be exponential and derive its mean. We define the area covered by a node while moving along a path as the total area that came under the node's coverage when it moved along that path. It must be noted that a node's coverage area at any point is the cell it is in at that point.

\begin{figure}[h]
  \centering
  \includegraphics[scale=0.55]{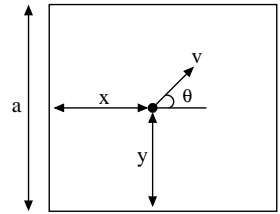} 
  \caption{Node performing RD motion}
  \label{rdmotion}
  \end{figure} 
\begin{lemma}
 The average area $\overline{A}$ covered during an epoch\footnote{Under RD mobility, a node chooses its mobility parameters such as speed and angle of direction, and travels with the chosen mobility parameters for a duration of time drawn from an exponential distribution (this duration of time is called an epoch). After that, the node halts for a random amount of time, following which it chooses a new set of values for its mobility parameters and the process repeats.}

 under RD mobility over a square-tiled terrain is given by $\overline{A} = 2a^{2} + \frac{4\overline{L}a}{\pi}$, where $a$ is the length of a side of a cell, and $\overline{L}$ is the average distance covered in an epoch of  the RD model.  
\end{lemma}

\begin{proof}
 Consider a node starting an epoch within a cell at the coordinate $(x, y)$ as illustrated in Fig~\ref{rdmotion}. Since the RD model has a uniform stationary distribution, these coordinates $x$ and $y$ are drawn from a uniform distribution over $[0,a]$. Also, according to the RD model, the node chooses an angle $\theta$ drawn from a uniform distribution over $[0,2\pi]$, velocity $v$ drawn from a uniform distribution over $[V_{min},V_{max}]$, duration of epoch $t$ drawn from an exponential distribution with average $\frac{\overline{L}}{\overline{v}}$, and a halting time at the end of the epoch drawn from a uniform distribution with average $\overline{T_{s}}$. Let the number of vertical walls and horizontal walls crossed during this epoch be $\eta_{v}$ and $\eta_{h}$, respectively. We have for $\theta \in [0,\frac{\pi}{2}]$,

\begin{equation*}
 \eta_{v} = \left\lceil \frac{vt\cos(\theta) - (a - x)}{a} \right\rceil
\end{equation*}
and 
\begin{equation*}
 \eta_{h} = \left\lceil \frac{vt\sin(\theta) - (a - y)}{a} \right\rceil.
\end{equation*}

We now derive the area covered in the epoch, by counting the number of cells that are traversed in that epoch. For most starting configurations, the node does not pass through corner-points of cells during the epoch. Hence, it is reasonable to assume that every wall crossed during the epoch is intersected somewhere along the side and therefore that each intersection represents crossover of the node into a new cell.  Hence the number of cells $n$ covered in the epoch for $\theta \in \left[0,\frac{\pi}{2}\right]$ is given by,

\begin{eqnarray*}
n & = & 1 + \eta_{v} + \eta_{h}\\
  & \leq & 3 + \frac{vt\cos(\theta) - (a - x)}{a} +\\
  && \quad \frac{vt\sin(\theta) - (a - y)}{a}\\
  & = & 1 + vt\frac{\left(\cos(\theta) + \sin(\theta)\right)}{a} + \frac{x+y}{a}.    
\end{eqnarray*}

Because of symmetry, the average number of cells covered in an epoch for $\theta \in \left[0,\frac{\pi}{2}\right]$ is the same as the average for $\theta \in \left[0,2\pi\right]$. Hence, the average number of cells $\overline{n}$ covered in an epoch is given by~Eq.~\ref{lengthyeqtn}, which on solving gives,
\begin{equation}
 \overline{n} = 2 + \frac{4\overline{L}}{\pi a}.
\end{equation}
Hence we have
\begin{equation*}
 \overline{A}  =  \overline{n}a^{2}
 =  2a^{2} + \frac{4\overline{L}a}{\pi}.
\end{equation*}
\end{proof}
Having estimated the average area covered during each epoch, we can proceed as in~\cite{methodology} to show that the expected hitting time\footnote{The hitting time is the time taken for a node starting from the stationary distribution to hit a randomly chosen static target.} $\overline{T_{h}}$ is as follows:

\begin{equation}
 \overline{T_{h}} = \frac{N}{2a^{2}+\frac{4\overline{L}a}{\pi}}\left( \frac{\overline{L}}{\overline{v}} + \overline{T_{s}}\right)
\end{equation}
In a similar way, as done in~\cite{methodology}, we can show that the meeting time $\overline{T_{m}}$ is exponentially distributed with an average given by
\begin{equation}
 \overline{T_{m}} = \frac{\overline{T_{h}}}{p_{m}v_{rd} + 2(1-p_{m})}
\end{equation}
where $v_{rd}$ is a constant~($1.27$) for the RD model such that the expected relative speed between two nodes $\overline{v_{rel}}$ is given by
\begin{equation}
 \overline{v_{rel}} = v_{rd}\overline{v},
\end{equation}
and $p_{m}$ is the probability that a node is moving at any point in time given by,
\begin{equation}
 p_{m} = \frac{\overline{T}}{\overline{T}+\overline{T_{s}}}.
\end{equation}
Finally we can show that the intermeeting time is also approximately exponentially distributed such that its average $\overline{T_{im}}$ equals $\overline{T_{m}}$.

\subsubsection{Contact Time~(CT)}
The contact time between two nodes is defined as the time elapsed between when the two nodes come within communication range of each other and when they go out of the communication range of each other. In this section, we derive the expression for the expected contact time ($\overline{T_{con}}$). We assume that the probability of contact between two nodes lasting over multiple cells is negligible. When a contact occurs, two cases arise: either one node could be stationary or both nodes could be moving. Let the expected contact time in the first case be $\overline{T_{con}^{s}}$ and in the second case be $\overline{T_{con}^{m}}$.

\begin{lemma}
\begin{equation*}
 \overline{T_{con}} = (1 - p_{m})\overline{T_{con}^{s}} + p_{m}\overline{T_{con}^{m}} 
\end{equation*}
\end{lemma}

\begin{proof}
 The probability that both nodes are moving given that one node is moving is $p_{m}$ and the probability that one node is stationary given that one node is moving is $1-p_{m}$. The result follows from this.
\end{proof}

\begin{lemma}
\begin{equation*}
 \overline{T_{con}^{s}} = 2I_{aw} + I_{ow}
\end{equation*}
where
\begin{equation*}
I_{aw} = \frac{Ka^{2}}{2}log\left(\frac{V_{max}}{V_{min}}\right)\left[log\left(\sqrt{2} + 1\right) + (\sqrt{2} - 1)\right]
\end{equation*}
and
\begin{equation*}
I_{ow} = 2Ka^{2}log\left(\frac{V_{max}}{V_{min}}\right)\left[log\left(\sqrt{2} + 1\right) + (1 - \sqrt{2})\right]
\end{equation*}
for $K$ given by
\begin{equation*}
K = \frac{1}{a\pi (V_{max} - V_{min})} 
\end{equation*}

\end{lemma}

\begin{proof}
 When one node is stationary, the contact begins when the other node enters the cell and ends when the other node leaves (our assumption that contacts last only over a cell is especially valid in this case). Hence the expected contact time in this case $\overline{T_{con}^{s}}$ is given by the expected time elapsed between the moving node entering and leaving the cell. We assume the moving node enters through any particular wall of the cell (since the problem is symmetric w.r.t the different walls of the cell) with a uniform distribution on point of entry, velocity, and angle of entry. Since epoch lengths are assumed large, it is assumed that with high probability, the movement of the node through the cell is part of one epoch and hence a straight line. The node can leave either through one of the two adjacent walls or the opposite wall. $I_{aw}$ is the product of the expected time taken to leave through a particular adjacent wall with the probability that the node leaves through that adjacent wall. Similarly, $I_{ow}$ is the product of the expected time taken to leave through the opposite wall with the probability that the node leaves through the opposite wall. Hence,
\begin{eqnarray}
 \overline{T_{con}^{s}}  =  I_{aw} + I_{aw} + I_{ow}
  =  2I_{aw} + I_{ow}.
\end{eqnarray}
  
$I_{aw}$ and $I_{ow}$ can be derived using geometric arguments which we do not present here due to space constraints.
\end{proof}

Now we approximate the contact time distribution with one node stationary as an exponential distribution with expected value $\overline{T_{con}^{s}}$. Hence the contact time for both nodes moving is given by the minimum of the times spent by each node in the cell and is hence approximately exponential with expected value $\frac{\overline{T_{con}^{s}}}{2}$. Therefore 
\begin{equation*}
 \overline{T_{con}^{m}} = \frac{\overline{T_{con}^{s}}}{2}.
\end{equation*}

\section{Modeling Contention under M2M Communication}
In this section we develop a framework for modeling contention in a DTN employing M2M communication. We first derive the distribution for the number of neighbors of any node and then derive the expected number of messages in the buffer of a node. These results are then used to develop the contention model.

\subsection{Distribution of Number of Neighbors}\label{sec:NeighboursDistribution}
We have established that the IMT and CT are approximately exponential with means $\overline{T_{im}}$ and $\overline{T_{con}}$, respectively. We hence model the process of nodes arriving within communication range of a particular node ($i.e.,$~entering the same cell as the node), and moving out of communication range of the node as a discrete time Markov chain, where the states are given by the number of neighbors at the beginning of each communication session. We assume that the probability of having more than $N_{max}$ neighbors is negligible. We also assume that since the duration of a communication session~($\tau$) is small, the probability of more than one arrival or departure in the duration of a communication session is negligible. Hence the transition probability $P(n+1|n)$ of the Markov chain is the probability of one arrival and no departure and is given by, 

\begin{equation*}
 P(n+1|n) =  \left(\frac{\lambda_{n}^{eff} \tau}{1!}e^{-\lambda_{n}^{eff} \tau}\right)\left(1 - \frac{\mu_{n}^{eff} \tau}{1!}e^{-\mu_{n}^{eff} \tau}\right)
\end{equation*}
where
\begin{equation*}
 \lambda_{n}^{eff} = \frac{(n_{tot} - n)}{\overline{T_{im}}}
\end{equation*}
and 
\begin{equation*}
 \mu_{n}^{eff} = \frac{n}{\overline{T_{con}}}.
\end{equation*}
In a similar way, $P(n-1|n)$ is given by
\begin{equation*}
 P(n+1|n) =  \left(1 - \frac{\lambda_{n}^{eff} \tau}{1!}e^{-\lambda_{n}^{eff} \tau}\right)\left(\frac{\mu_{n}^{eff} \tau}{1!}e^{-\mu_{n}^{eff} \tau}\right)
\end{equation*}
and $P(n|n)$ is given by,
\begin{equation*}
 P(n|n) = 1 - P(n+1|n) - P(n-1|n).
\end{equation*}
The steady state distribution of this Markov chain gives us the probability distribution of neighbors. We assume in the subsequent sections that the probability that the number of neighbors $N_{nhb}$ equals $j$ is denoted by $v_{j}$. 
\subsection{Expected Number of Messages in Buffer}
In order to derive the expression for the expected number of messages in buffer ($E[B]$), we note that by Little's Law, we have,
\begin{equation}
 E[B] = \lambda E[W]
\end{equation}
where $\lambda$ is the total traffic arrival rate at each node (assumed Poisson) and $E[W]$ is the expected waiting time for a message in a node's buffer.
 
We first derive the expression for the arrival rate $\lambda$. We then derive $E[W]$ in the next section. $\lambda$ is given by
\begin{equation*}
 \lambda = \lambda_{gen} + \lambda_{rel}
\end{equation*}
where $\lambda_{gen}$ is the rate of generation of messages at that node and $\lambda_{rel}$ is the rate at which relay messages (messages for which this node is going to act as relay) arrive at the node. To calculate $\lambda_{rel}$, we assume that any given message spreads to nearly the entire network before being removed by the recovery mechanism. This is valid for the direct recovery mechanism we assume is used, which makes minimal effort at recovery. Hence, we assume that any message in the network is replicated at a node with a high probability before all its copies are erased from the network. Now, working under this assumption, we treat the network as a queue with messages generated at other nodes considered as arrivals (with an average rate $\lambda_{gen}(n_{tot} - 1)$), and messages delivered to a particular node under consideration as service completions (with an average rate $\lambda_{rel}$). Clearly, for the number of messages in the network to be stable, the service rate $\lambda_{rel}$ must be at least $\lambda_{gen}(n_{tot} - 1)$). We approximate the value of $\lambda_{rel}$ by this lower bound.

\subsubsection{Expected Waiting Time for Epidemic Routing}
The expected waiting time depends on several factors such as the chosen mobility model, routing protocol, and the recovery mechanism. In order to derive the expected waiting time for epidemic routing, we assume direct recovery is used, $i.e.,$~a node will only drop a message after meeting the destination and either transmitting the message to the destination itself or learning from the destination node that the destination node has already received the message. We also assume that the scheduler at each node works such that whenever the destination of some messages is encountered, those messages are given top priority. We also assume that if contact is established with the destination of some messages in the buffer, the bandwidth available is sufficient to transmit all the messages to the destination. This is valid since we assume symmetric traffic generation, $i.e,$~for any source node, there is an equal amount of traffic generated destined for every other node and hence a limited amount of traffic destined for any particular destination. Under these assumptions we consider a particular message in the buffer of a particular node (call it the current node) and derive its expected waiting time for epidemic routing. Let $p_{succ}$ be the probability of successfully communicating with the destination of the message given that the current node meets the destination~(successful communication may not happen because of contention with other neighboring nodes). The number of meetings with the destination before the current node is successfully able to communicate with it is geometrically distributed with mean $\frac{1}{p_{succ}}$ and the average time elapsed between meetings is $\overline{T_{im}}$. Hence, $E[W]$ is given by
\begin{equation}
E[W] = \frac{\overline{T_{im}}}{p_{succ}}
\end{equation}
where $p_{succ}$ is given by
\begin{equation*}
p_{succ} = 1 - \left(1 - p_{succ}^{'}\right)^{\frac{\overline{T_{con}}}{\tau}} 
\end{equation*}
where $\frac{\overline{T_{con}}}{\tau}$ is the average number of communication sessions for which the current node remains in contact with the destination, and $p_{succ}^{'}$ is the probability that the current node successfully communicates with the destination in a particular communication session. $p_{succ}^{'}$ is given by
\begin{equation*}
 p_{succ}^{'}  =  \sum_{j=1}^{\infty}P\left(N_{nhb}=j|N_{nhb}\geq 1\right) (p_{comm}^{j})^2
\end{equation*}
where $p_{comm}^{j}$ is the probability that a particular node gets to communicate in a session given that it has $j$ neighbors. $p_{comm}^{j}$ depends upon the MAC protocol used to resolve contention. As described in Section~\ref{sec:MACProtocol}, the MAC protocol effectively chooses all nodes in a cell if the number of nodes is atmost $\alpha$, else it chooses a random $\alpha$-sized subset of the nodes in the cell. Hence,
\begin{multline*}
 p_{succ}^{'} = \sum_{j=1}^{\infty}\frac{v_{j}}{\left(\sum_{k=1}^{\infty}v_{k}\right)}\min\left(\frac{\alpha}{j+1},1\right)\min\left(\frac{\alpha}{j+1},1\right)
\end{multline*}

\subsection{Contention Model}\label{sec:ContentionModel}
We model contention in terms of three different factors as follows:~$(i)$~the probability of increase in number of copies of a particular message from $i$ to $i^{'}$, namely $p_{i^{'},i}$,~$(ii)$ the probability of delivery to the destination of a particular message given that there are $i$ copies of the message in the network, namely $p_{d,i}$, and~$(iii)$~the expected time lapse between the number of copies reaching $i$ and either a rise in the number of copies or delivery to the destination, namely $E[D_{i}]$. We assume that with a high probability, in any communication session, only at most one of the $i$ nodes carrying a copy of the message attempts to transmit the message and hence $p_{i^{'},i}$ is zero for $i^{'} \geq i+\alpha$.

For $i< i^{'} < i+\alpha$, since we assume that at most one node transmits this message in any communication session,

\begin{multline*}
 p_{i^{'},i} = i\sum_{j=i^{'}-i}^{\alpha-1}P(N_{nhb}=j)\times\min\left(\frac{\alpha}{j+1},1\right)\times\\
\left(p_{tx}(i)\right)^{i^{'}-i}\times\left(1-p_{tx}(i)\right)^{j-(i^{'}-i)}
\end{multline*}
where $P(N_{nhb} = j)$ has already been derived in Section~\ref{sec:NeighboursDistribution}, and $p_{tx}(i)$ is the probability that a message is successfully transmitted to another node given they are in communication with each other and that there are $i$ copies in the network. The probability $p_{tx}(i)$ is given by,

\begin{equation*}
 p_{tx}(i) = \left(1-\frac{i-1}{n_{tot}-1}\right)\left(\frac{b_{BW}}{E[B]}\right).
\end{equation*}

To compute $p_{d,i}$, we assume that at most one of the $i$ nodes carrying a copy of the message meet the destination at a given point in time. Hence $p_{d,i}$ is the probability that any one of these nodes meets the destination and successfully transfers the message to the destination. And $p_{d,i}$ is given by,

\begin{equation*}
\begin{split}
 p_{d,i} =& \sum_{j=1}^{\infty}P(N_{nhb}=j)\times\frac{j}{n_{tot}-1}\times\\
 &\quad\min\left(\frac{\alpha}{j+1},1\right)\times
\min\left(\frac{\alpha}{j+1},1\right).
 \end{split}
\end{equation*}

\subsubsection{$E[D_{i}]$ Computation}\label{sec: ediDerivation}
To derive $E[D_{i}]$, we first estimate the average time elapsed before a particular node out of the $i$ nodes carrying the message spreads a copy/copies of the message. We jointly model the process of arrival and departure of neighbors and the attempted transfer of the message as a stochastic process with states $j = 0, 1, \cdots n_{max}$ representing the event that at the beginning of a communication session, the node has $j$ neighbors and still has not spread copies of the message. The stochastic process also has a state $C$ representing that the node has successfully spread one or more copies of the message and that therefore the number of message copies in the network has changed. This has been illustrated in Fig~\ref{nsystem}. 
\begin{figure}[h]
  \centering
  \includegraphics[scale=1]{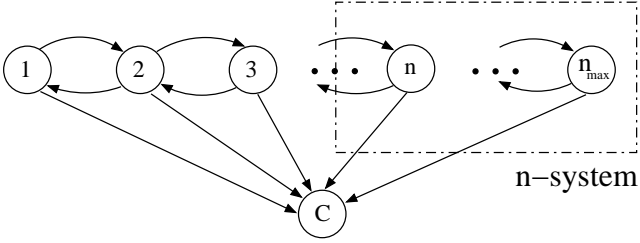} 
  \caption{Digraph representing the stochastic process used to compute $E[D_i]$}
  \label{nsystem}
  \end{figure}

We define the $n$-system as the subset of states $n, n+1 \cdots, n_{max}$. Let $E[t_{n}]$ be the expected dwell time in the state $n$, $p_{n,n+1}^{\overline{s}}(i)$ be the probability that the number of neighbors of the node under consideration increase from $n$ to $n+1$ in a communication session with no message copies spread, given there are $i$ copies of the message, $n$ neighbors for that particular node, and that some change in the state of the stochastic process happens at the end of the communication session. Also, let $p_{s(n+1),n}^{\overline{s}}(i)$ be the probability that a transition happens from the $(n+1)$-system to the state $n$ given that a transition from the $(n+1)$-system occurs. The following lemma gives the expected dwell time $E[t^{n}]$ of the process in the $n$-system.
\begin{lemma}
 \begin{equation*}
  E[t^{n}] = \frac{E[t_{n}] + p_{n,n+1}^{\overline{s}}(i)\times E[t^{n+1}]}{1 - p_{n,n+1}^{\overline{s}}(i)\times p_{s(n+1),n}^{\overline{s}}(i)}
 \end{equation*}
\end{lemma}

\begin{proof}
 We replace the dwell time distribution at each state $n$ by its expected value $E[t_{n}]$. Then the expected time spent in the $n$-system is given by,
\begin{equation*}
\begin{split}
 E[t^{n}] = & E[t_{n}] + \\
&p_{n,n+1}^{\overline{s}}(i)\left(E[t_{n+1}] + p_{s(n+1),n}^{\overline{s}}(i)\times \left(E[t_{n}] + \cdots\right)\right)
\end{split}
\end{equation*}
which gives us
\begin{equation*}
 E[t^{n}] = E[t_{n}] + p_{n,n+1}^{\overline{s}}(i)\left(E[t_{n+1}] + p_{s(n+1),n}^{\overline{s}}(i)\times E[t^{n}]\right)
\end{equation*}
which can be rearranged to give us the result.
\end{proof}
To terminate the recursion, we have
\begin{equation*}
 E[t^{n_{max}}] = E[t_{n_{max}}].
\end{equation*}

To compute $E[t_{n}]$ we first compute $p_{n,n}^{\overline{s}}(i)$ which is the probability that the number of neighbors remains $n$ at the end of a communication session with no copies of the message transferred, given that there are $i$ copies of the message in the network. For $n \leq \alpha - 1$ it is easy to derive
\begin{equation*}
 p_{n,n}^{\overline{s}}(i) = P(n|n)\times p_{n}^{\overline{s}}(i)
\end{equation*}
where $P(n|n)$ is as derived in the aforementioned Section~\ref{sec:NeighboursDistribution} and $p_{n}^{\overline{s}}(i)$ is the probability that there is no successful communication of the message given that there are $i$ copies of the message and $n$ neighbors. For $n \leq \alpha -1$, $p_{n}^{\overline{s}}(i)$ is given by,
\begin{equation*}
 p_{n}^{\overline{s}}(i) = \left(1 - p_{tx}(i)\right)^{n}.
\end{equation*}
 And for $n > \alpha - 1$,  $p_{n}^{\overline{s}}(i)$ is given by,
\begin{equation*}
 p_{n}^{\overline{s}}(i) = \left(1 - \frac{\alpha}{n} + \frac{\alpha}{n}\times \left(1 - p_{tx}(i)\right)^{n}\right).
\end{equation*}
The number of communication sessions for which the system stays in state $n$ is then a geometric random variable with probability of success $1 - p_{n,n}^{\overline{s}}(i)$. Hence, the expected dwell time in state $n$ is given by,
\begin{equation*}
 E[t_{n}] = \frac{\tau}{1 - p_{n,n}^{\overline{s}}(i)}.
\end{equation*}
Also, it must be noted that since the session time is small, the expected dwell times $E[t_{n}]$ can be approximated by continuous exponential distributions. Hence, if we considered $i$ nodes in state $n$ instead of one and computed the expected dwell time (the dwell time in this case would be the minimum of the dwell times of the $i$ nodes), the expected dwell time would approximately be $\frac{E[t_{n}]}{i}$.   

The transition probability $p_{n,n+1}^{\overline{s}}(i)$ is given by,
\begin{equation*}
 p_{n,n+1}^{\overline{s}}(i) = \frac{p_{n}^{\overline{s}}(i)\times P(n+1|n)}{1 - p_{n}^{\overline{s}}(i)\times P(n|n)}. 
\end{equation*}
To compute the other transition probability $p_{s(n+1),n}^{\overline{s}}(i)$, we assume that the dwell time within the $(n+1)$-system is sufficient for the stationary distribution within the $(n+1)$-system to be reached. Hence we have,
\begin{equation*}
 p_{s(n+1),n}^{\overline{s}}(i) = \frac{v_{n+1}\times P(n|n+1)}{v_{n+1}\times P(n|n+1) + \sum_{k = n+1}^{n_{max}}v_{k}\times (1 - p_{k}^{\overline{s}}(i))}.
\end{equation*}

The expected time elapsed before a particular node carrying the message spreads copies of the message is finally given by $E[t^{0}]$. However, $E[D_{i}]$ is given by the expected time elapsed before any one of the $i$ nodes carrying the message spreads copies of the message. Since each $E[t_{n}]$ is approximately inversely proportional to the number of nodes we consider for the stochastic process, and $E[t^{0}]$ is a linear combination of $E[t_{n}]$, we can approximately say

\begin{equation*}
 E[D_{i}] = \frac{E[t^{0}]}{i}.
\end{equation*}
   
\section{Routing Delay}
In this section we derive the routing delay for epidemic routing in a DTN which uses M2M communication. We model the evolution of the network with respect to a specific message, from the generation of the message to the spreading of copies of the message to finally the delivery of the message, as a stochastic process, represented as a digraph in Fig~\ref{delaystochprocess}. This stochastic process has a set of states $i={1,\cdots ,n_{tot}-1}$ where state $i$ corresponds to the presence of $i$ copies of the message in the network, and $n_{tot}$ refers to the total number of nodes in the network. The stochastic process also has an absorbing state $D$ which represents successful delivery of the message to the destination. The network spends an average time of $E[D_{i}]$ in state $i$ before moving on to one of the next possible states. Since we assume no message drops ($i.e.,$~buffers are adequately sized), the probability of transition $p_{i^{'},i}$ from state $i$ to $i^{'}$ is zero for $i^{'}<i$. 
\begin{figure}[h]
  \centering
  \includegraphics[scale=1]{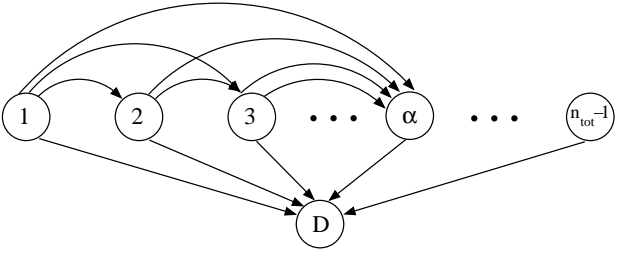} 
  \caption{Digraph representing the stochastic process used to compute the routing delay}
  \label{delaystochprocess}
  \end{figure}
When the network is in state $i$, we assume that the probability that more than one of the nodes carrying a particular message try to transmit the message at the same time is negligible. Hence, when the network is in state $i$, at most one of the $i$ nodes will transmit the message. This can be transmitted to at most $\alpha-1$ other nodes (since at most $\alpha$ nodes can participate in an M2M communication session). Hence, the network can evolve from state $i$ only to states $i+1,\cdots, i+\alpha-1$. We call the set of states reachable in a single transition from state $i$ as $Ch(i)$. This corresponds to the children of node $i$ in the digraph representing the stochastic process (assuming edges corresponding to zero-probability transitions are deleted). Similarly, a state $i$ can be entered only from states $i-1, \cdots, i-\alpha+1$. We call the set of states from which state $i$ is reachable in a single transition as $Pa(i)$. This corresponds to the parents of node $i$ in the digraph.

We then derive the expected delay $E[D]$ of epidemic routing by averaging the delay over all the possible paths $\Phi$ on the digraph starting from the starting state $1$ and ending in state $D$. It must be noted such a path corresponds to a sequence of states $s_{1}, s_{2}, \cdots, s_{m}$ where $s_{1}=1$ and $s_{m}=D$.

\begin{eqnarray*}
 E[D] & = & \sum_{\Phi} P\left(\Phi\right)D\left(\Phi\right)\\
      & = & \sum_{\left(1,\cdots,D\right)}\left(P(s_{2}|1)\cdots P(D|s_{m-1})\right)\times\\
      &&\quad\left(E[D_{1}] + \cdots E[D_{m-1}]\right).
\end{eqnarray*}
This expression can be re-written as follows:
\begin{multline*}
 E[D] = \sum_{s}\sum_{p\in Pa(s)}\sum_{c\in Ch(s)\cup\{D\}} \\
P(Path(1,p)).p_{s,p}.p_{c,s}.P(Path(c,D)).E[D_{s}]
\end{multline*}
where $P(Path(i,j))$ represents the total probability of all paths starting from state $i$ and ending at state $j$, given that the system is initially in state $i$. $P(Path(1,i))$ is given by the following recursion
\begin{multline*}
 P(Path(1,i)) = \sum_{p\in Pa(i)}P(Path(1,p)).p_{i,p}.
\end{multline*}
And this can be calculated for all states $i$ efficiently using dynamic programming. Similarly, $P(Path(i,D))$ is given by the recursion
\begin{multline*}
 P(Path(i,D)) = \sum_{c\in Ch(i)}\left(p_{c,i}.P(Path(c,D))\right) + p_{D,i}.
\end{multline*}
And this can also be calculated efficiently for all $i$ using dynamic programming.
\begin{table}
\centering
\renewcommand{\arraystretch}{1.5}
\caption{Simulation settings}
\label{settings}
\begin{tabular}{|l|c|}\hline
{\bf Parameters}&{\bf Values}\\\hline
Terrain size (in m$\times$m)&$100\times100$\\\hline
Grid size (in m)&$10\times 10$\\\hline
M2M communication session time (in seconds) &$0.1$\\\hline
M2M $\alpha$-value & $4$\\\hline
Message size (in Kb)&$1$\\\hline
Message generation rate &{$\frac{1}{5}$}\\\hline
Message TTL&$\infty$\\\hline
Buffer size &$\infty$\\\hline
Scheduling&FCFS\\\hline
Simulation time (in seconds)&$10000$\\\hline
\end{tabular}
\end{table}

\section{Performance Evaluation}
The simulations were done using the ONE simulator~\cite{onesim}, which we custom modified to simulate one-to-one CSMA and M2M communication schemes. The simulation settings are shown in Table~\ref{settings}. Each experiment was run for $10000$ simulation seconds. The system was verified to have reached the stability region within this time. As mentioned earlier, traffic generation is symmetric, $i.e.,$~every message generated is given a randomly chosen destination node. 

\begin{figure}
 \centering \subfloat{\label{2c}\includegraphics[scale=0.3]{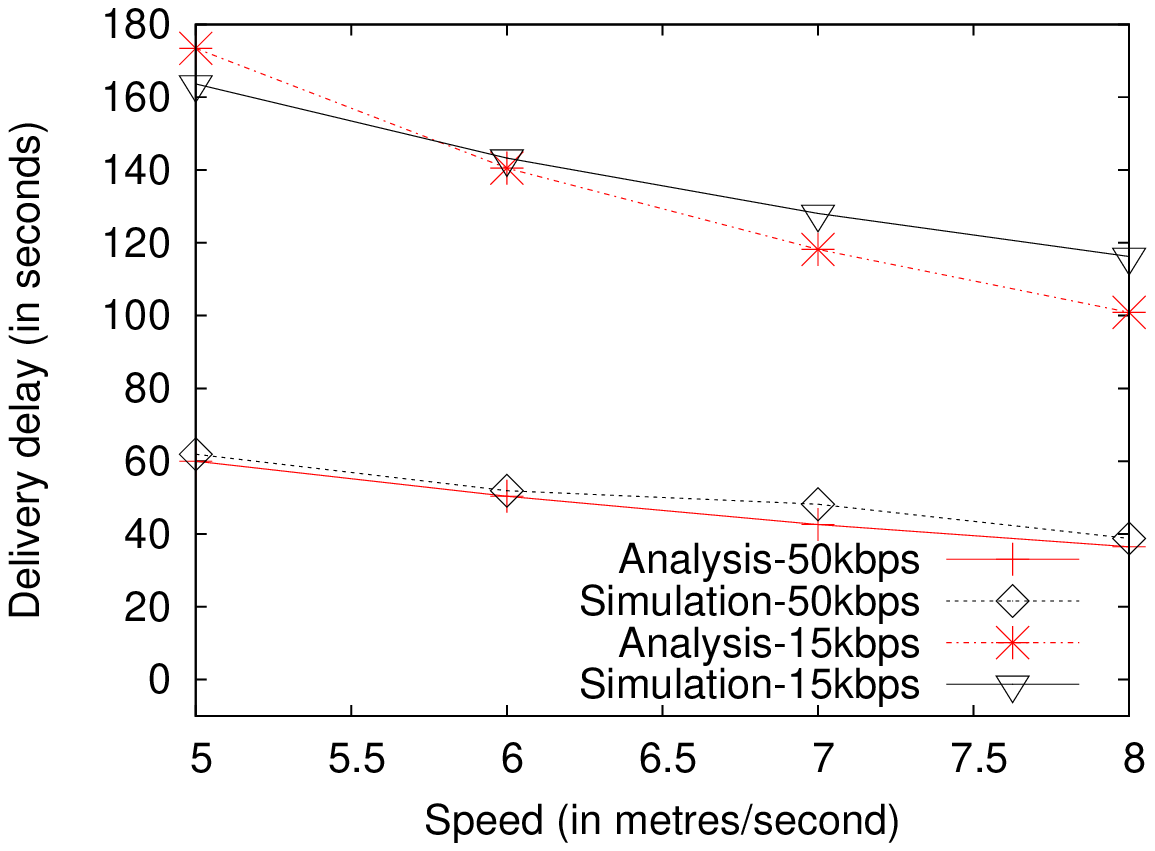}}
\subfloat{\label{4c}\includegraphics[scale=0.3]{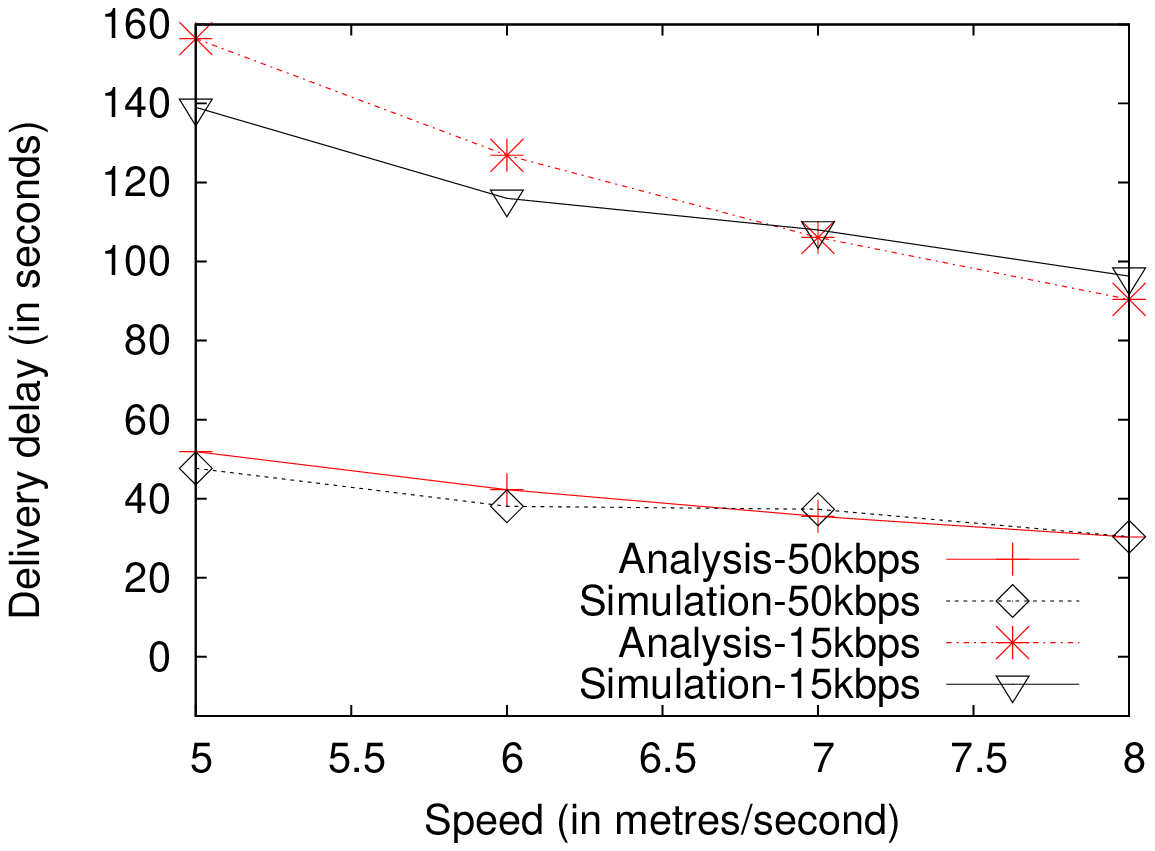}}
\caption{Delivery delay performance over two different number of nodes (namely, $10$ and $20$) and across different average speeds~($\overline{v}$), for two different link bandwidths}
\label{delayperformance} 
\end{figure}
\begin{figure*}
 \centering \subfloat{\label{2c}\includegraphics[scale=0.3]{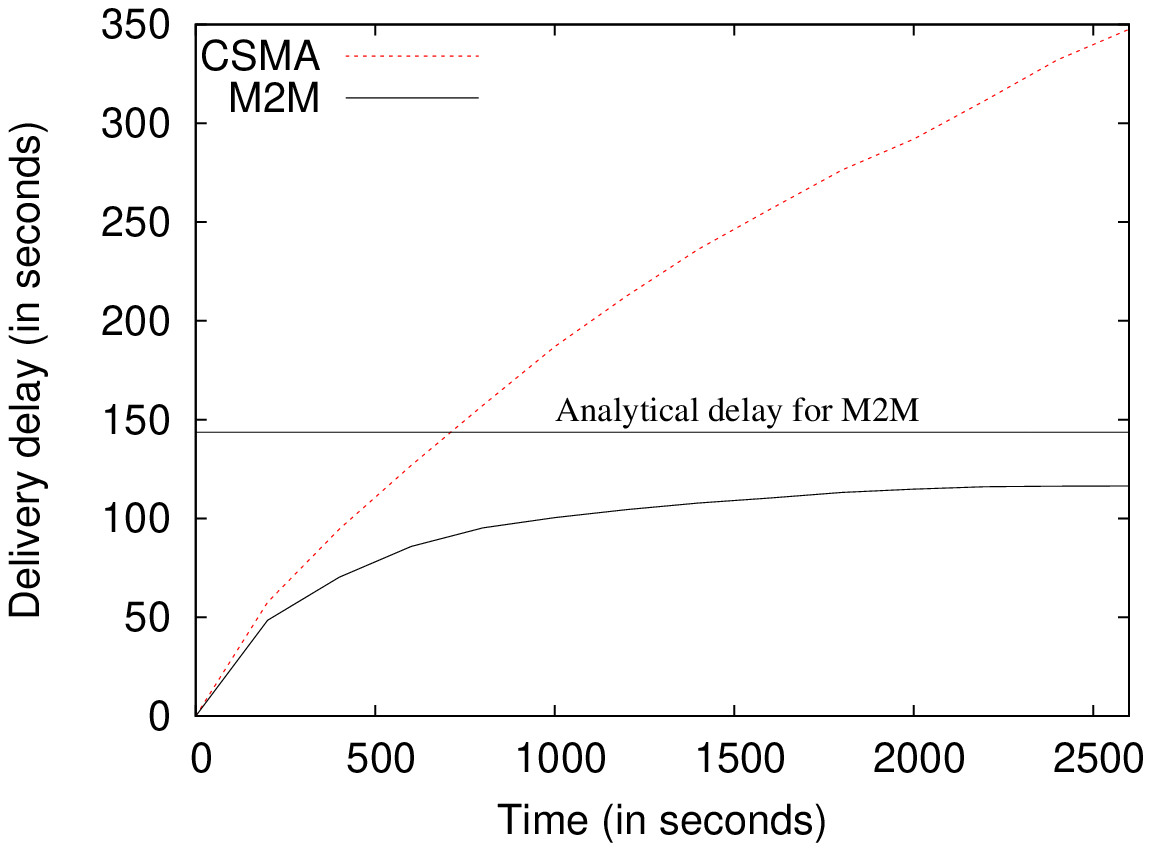}}
\subfloat{\label{4c}\includegraphics[scale=0.3]{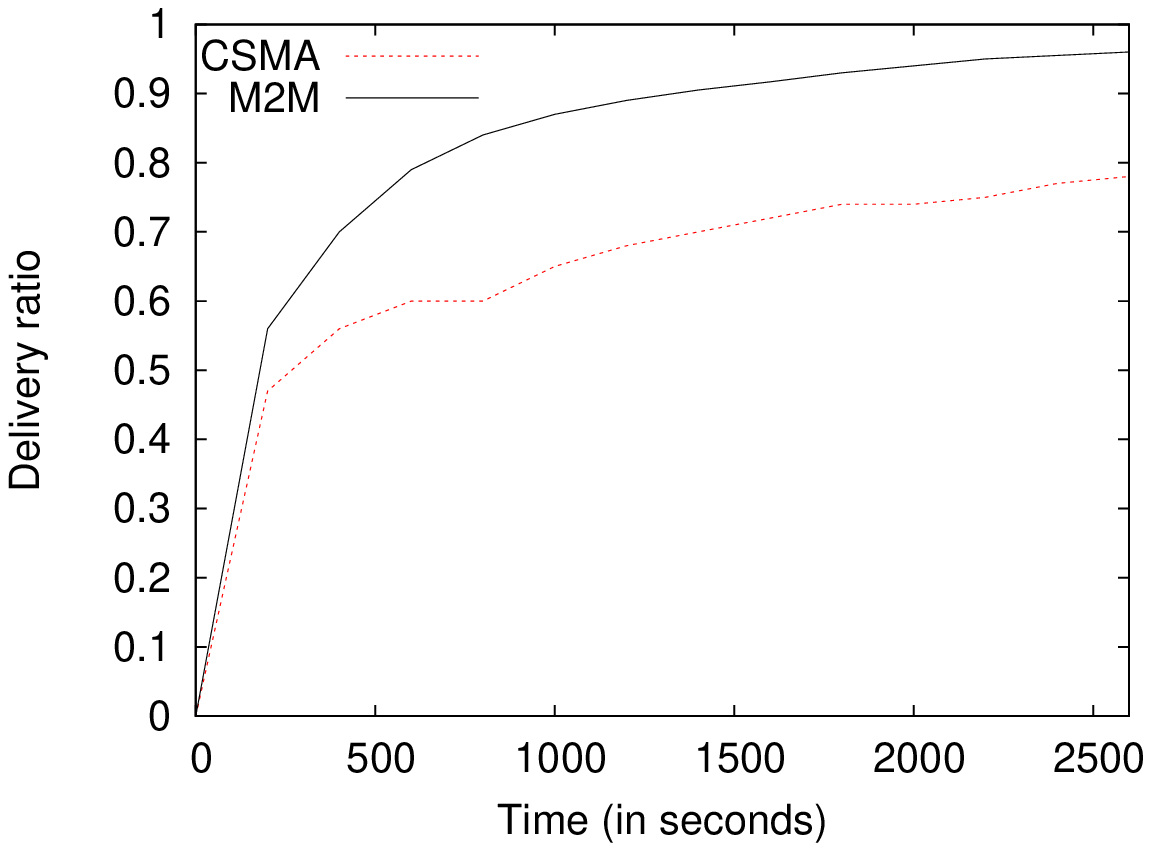}}
\subfloat{\label{2c}\includegraphics[scale=0.3]{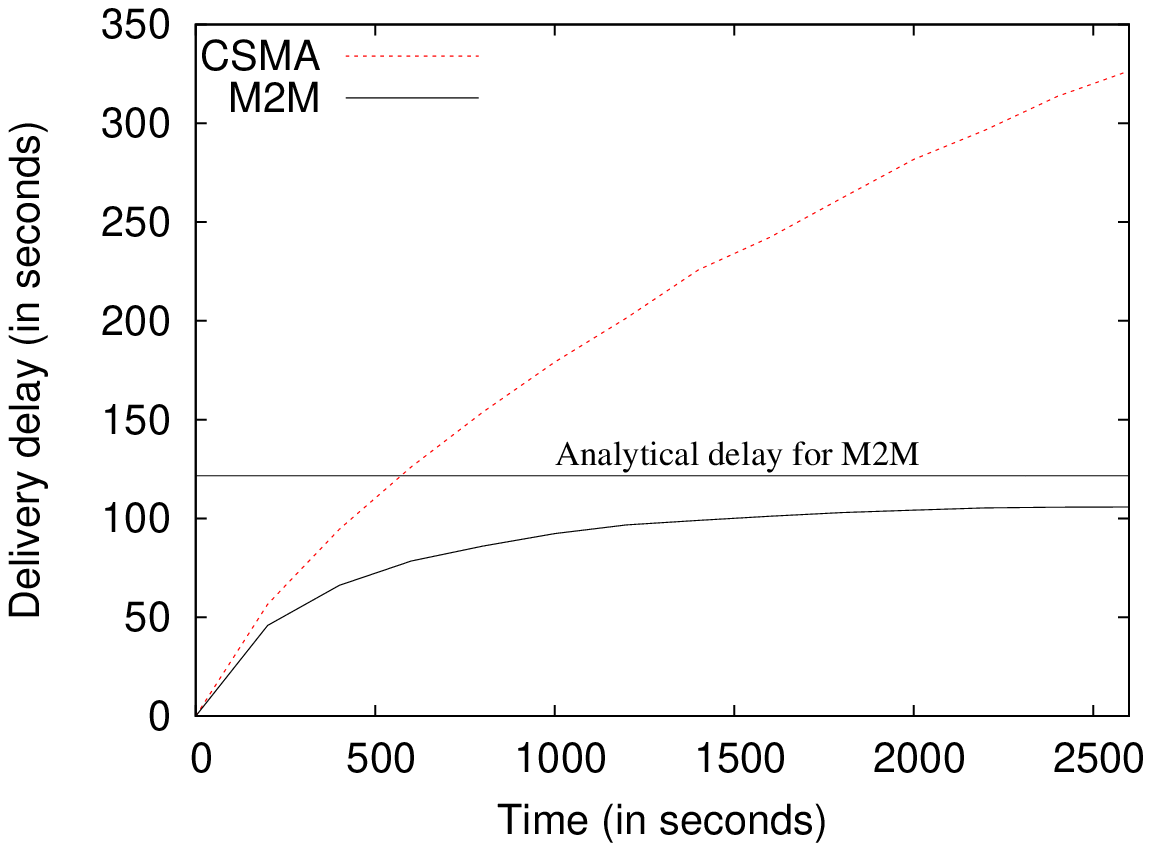}}
\subfloat{\label{4c}\includegraphics[scale=0.3]{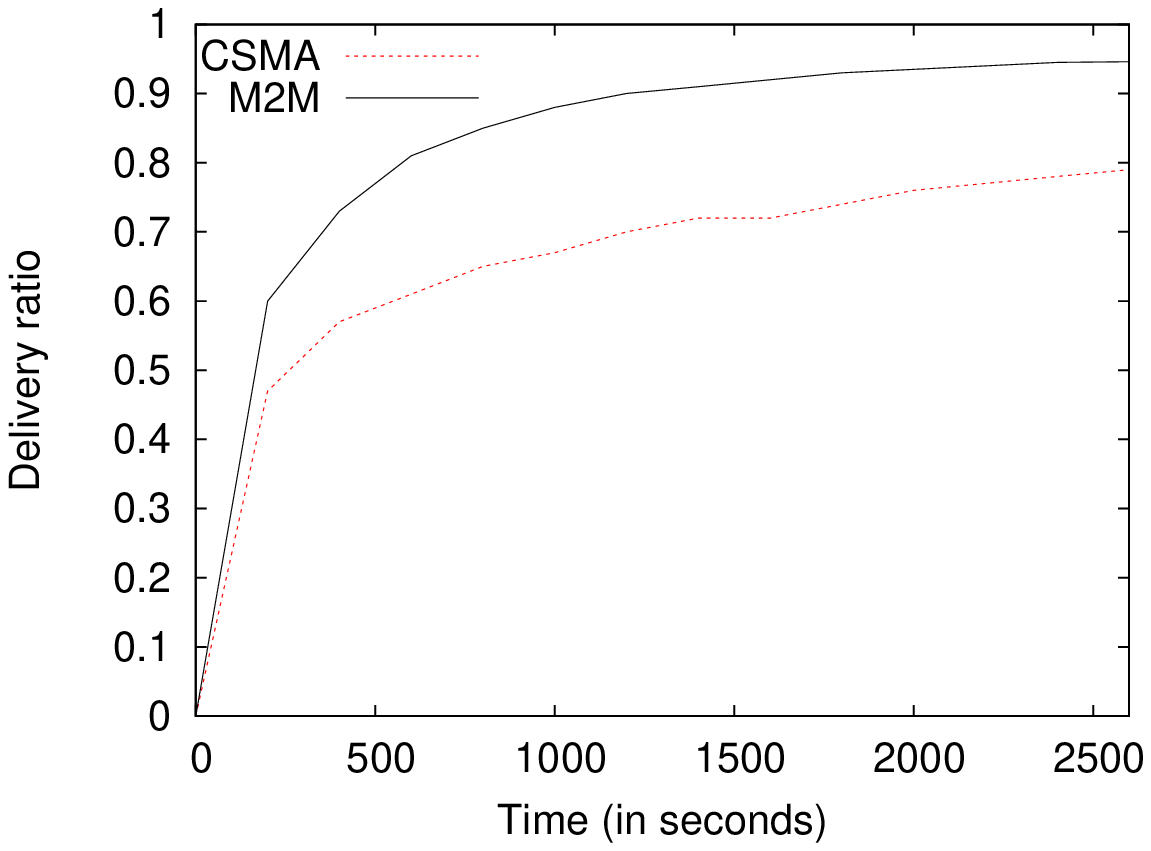}}

\caption{Routing performance of CSMA and M2M communication schemes for $100$ nodes with two different average node speeds~($\overline{v}$), namely $7$ and $8$ m/sec}
\label{100stability} 
\end{figure*}

Figure~\ref{delayperformance} compares the routing delay values obtained from analysis and simulations for $10$ and $20$ nodes, and for two different link bandwidths, namely $15$~kbps and $50$~kbps. For each configuration, the comparison is done over average node speeds~($\overline{v}$) ranging from $5$~m/sec to $8$~m/sec. These delay values were observed in the stability region of the experiments and the delivery ratios ($i.e.,$~the fraction of messages delivered successfully) were observed to be close to $0.9$. While we have made a number of approximations in developing our analytical model, we find that the delay values from analysis and simulations match, thereby validating our model and showing that the approximations were reasonable ones to make. As is expected, it is observed that with an increase in node speeds the routing delay decreases, as the mixing of nodes enable faster delivery of messages. Moreover, an approximately three-fold increase in the link bandwidth proportionally reduces the routing delay, which is intuitive as messages are pumped three times faster for the same contact time scenarios. 

Figure~\ref{100stability} compares the routing performance of one-to-one CSMA and M2M schemes. These experiments were run with a lower link bandwidth of $10$~kbps for M2M and $40$~kbps for one-to-one CSMA. CSMA is given a higher bandwidth since it uses the entire spectrum available while M2M communication divides the spectrum among the $\alpha$ ($=4$) nodes participating in a communication session using FDMA. The M2M scheme clearly outperforms the one-to-one CSMA scheme. We find that when the M2M scheme reaches its stability region, the CSMA scheme has not yet attained stability, and the differences between the two schemes are increasing with time. The fact that M2M has reached its stability region while CSMA has not, is supported by the high delivery ratio values~($i.e.,$~above $0.9$) for M2M and lower delivery ratio values~($i.e.,$~around $0.8$) for CSMA. Also our M2M delay values stabilize near the corresponding analytical delay values.

\section{Conclusion}
In this paper, we have shown for the first time how M2M communication can be used to combat contention, and thereby reduce the routing delay in a DTN. A theoretical model has been developed for the delivery delay of a DTN employing epidemic routing and M2M communication and has been validated against simulations. Using simulations, we have also shown that M2M communication significantly outperforms communication that uses one-to-one CSMA in terms of delivery delay. In the future, we plan to extend the model developed to work with more sophisticated mobility models. 

\section*{Acknowledgement}
This work was supported by the Department of Science and Technology~(DST), New Delhi, India.

% \bibliographystyle{IEEEtran}
% \bibliography{m2m}

\begin{thebibliography}{10}
\providecommand{\url}[1]{#1}
\csname url@samestyle\endcsname
\providecommand{\newblock}{\relax}
\providecommand{\bibinfo}[2]{#2}
\providecommand{\BIBentrySTDinterwordspacing}{\spaceskip=0pt\relax}
\providecommand{\BIBentryALTinterwordstretchfactor}{4}
\providecommand{\BIBentryALTinterwordspacing}{\spaceskip=\fontdimen2\font plus
\BIBentryALTinterwordstretchfactor\fontdimen3\font minus
  \fontdimen4\font\relax}
\providecommand{\BIBforeignlanguage}[2]{{%
\expandafter\ifx\csname l@#1\endcsname\relax
\typeout{** WARNING: IEEEtran.bst: No hyphenation pattern has been}%
\typeout{** loaded for the language `#1'. Using the pattern for}%
\typeout{** the default language instead.}%
\else
\language=\csname l@#1\endcsname
\fi
#2}}
\providecommand{\BIBdecl}{\relax}
\BIBdecl

\bibitem{dtnarch}
V.~Cerf, S.~Burleigh, A.~Hooke, L.~Torgerson, R.~Durst, K.~Scott, K.~Fall, and
  H.~Weiss, ``{RFC} 4838-{D}elay-tolerant networking architecture,'' \emph{IRTF
  DTN Research Group}, 2007.

\bibitem{dtnsurvey}
M.~Khabbaz, C.~Assi, and W.~Fawaz, ``Disruption-tolerant networking: {A}
  comprehensive survey on recent developments and persisting challenges,''
  \emph{IEEE Communications Surveys Tutorials}, vol.~14, no.~2, pp. 607--640,
  2012.

\bibitem{contentionaware}
A.~Jindal and K.~Psounis, ``Contention-aware performance analysis of
  mobility-assisted routing,'' \emph{IEEE Transactions on Mobile Computing},
  vol.~8, no.~2, pp. 145--161, 2009.

\bibitem{m2mtwc09}
R.~de~Moraes, J.~Garcia-Luna-Aceves, and H.~Sadjadpour, ``Many-to-many
  communication for mobile ad hoc networks,'' \emph{IEEE Transactions on
  Wireless Communications}, vol.~8, no.~5, pp. 2388--2399, 2009.

\bibitem{m2minfocom}
R.~de~Moraes, H.~Sadjadpour, and J.~Garcia-Luna-Aceves, ``Many-to-many
  communication: {A} new approach for collaboration in {MANETs},'' in
  \emph{INFOCOM '07: Proceedings of the International Conference on Computer
  Communications}, 2007, pp. 1829--1837.

\bibitem{pqepidemic}
T.~Matsuda and T.~Takine, ``(p,q)-{E}pidemic routing for sparsely populated
  mobile ad hoc networks,'' \emph{IEEE Journal on Selected Areas in
  Communications}, vol.~26, no.~5, pp. 783--793, 2008.

\bibitem{spraynwait}
T.~Spyropoulos, K.~Psounis, and C.~S. Raghavendra, ``Spray and wait: An
  efficient routing scheme for intermittently connected mobile networks,'' in
  \emph{WDTN '05: Proceedings of the ACM SIGCOMM Workshop on Delay-Tolerant
  Networking}, 2005, pp. 252--259.

\bibitem{multicopy}
T.~Spyropoulos, K.~Psounis, and C.~Raghavendra, ``Efficient routing in
  intermittently connected mobile networks: {T}he multiple-copy case,''
  \emph{IEEE/ACM Transactions on Networking}, vol.~16, no.~1, pp. 77--90, 2008.

\bibitem{ode}
X.~Zhang, G.~Neglia, J.~Kurose, and D.~Towsley, ``Performance modeling of
  epidemic routing,'' \emph{Computer Networks}, vol.~51, no.~10, pp.
  2867--2891, 2007.

\bibitem{contentionaware-conf}
A.~Jindal and K.~Psounis, ``Performance analysis of epidemic routing under
  contention,'' in \emph{IWCWC '06: Proceedings of the International Conference
  on Wireless Communications and Mobile Computing}, 2006, pp. 539--544.

\bibitem{lambdard}
R.~Groenevelt, ``Stochastic models for mobile ad hoc networks,'' Ph.D.
  dissertation, University of Nice Sophia Antipolis, 2005.

\bibitem{methodology}
T.~Spyropoulos, A.~Jindal, and K.~Psounis, ``An analytical study of fundamental
  mobility properties for encounter-based protocols,'' \emph{International
  Journal of Autonomous and Adaptive Communication Systems}, vol.~1, no.~1, pp.
  4--40, 2008.

\bibitem{onesim}
A.~Ker\"{a}nen, J.~Ott, and T.~K\"{a}rkk\"{a}inen, ``The {ONE} simulator for
  {DTN} protocol evaluation,'' in \emph{SIMUTools '09: Proceedings of the
  International Conference on Simulation Tools and Techniques}, 2009, pp.
  55:1--55:10.

\end{thebibliography}

\end{document}